\documentclass[12pt,a4paper]{article}
\usepackage[dvipdfmx,hiresbb]{graphicx}
\usepackage[dvipdfmx]{color}

\usepackage{color}
\usepackage{bm,enumerate,amsmath,amssymb,amsthm}
\usepackage{epsfig}
\usepackage{cite}
\usepackage{algorithm}
\usepackage{algpseudocode}
\usepackage{txfonts}
\usepackage{subfigmat}
\usepackage{figsize}
\usepackage{here}
\usepackage{listings}
\usepackage{braket}
\usepackage{comment}
\usepackage{lscape}
\usepackage{bigints}

\usepackage{arxiv}
\usepackage[utf8]{inputenc} 

\theoremstyle{definition}
\newtheorem{theorem}{Theorem}
\newtheorem*{theorem*}{Theorem}
\newtheorem{definition}{Definition}
\newtheorem*{definition*}{Definition}

\newtheorem{lemma}{Lemma}

\newtheorem{remark}{Remark}

\newtheorem{assumption}{Assumption}

\renewcommand{\headeright}{}

\title{Lie-Group Mode Connectivity in Quantum Machine Learning from a Dynamical Lie Algebra Perspective}

\author{
  Hiroshi Ohno\\
  Toyota Central R\&D Labs., Inc.\\
  Aichi, Japan\\
  \texttt{oono-h@mosk.tytlabs.co.jp}\\
}

\date{\empty}

\begin{document}

\maketitle

\begin{abstract}
  Mode connectivity has been widely studied in classical machine learning as a geometric property of low-loss regions in parameter space.
  In quantum machine learning (QML), however, the physically relevant object is not the parameter vector itself but the unitary transformation implemented by a parameterized quantum circuit.
  In this study, we formulate mode connectivity on the reachable unitary Lie group generated by the dynamical Lie algebra of the generators.
  We show that, under a near-minimum connectedness assumption and the absence of critical values in a low-loss band, the corresponding low-loss sublevel set on the reachable Lie group is path-connected.
  This provides a geometric interpretation of mode connectivity in QML that is independent of a particular parameterization.
  We further discuss how overparameterization can enable the lifting of Lie-group paths to parameter space, thereby making Lie-group connectivity observable in parameter-space experiments.
  Finally, we present toy numerical experiments in which geodesic interpolations between trained unitaries exhibit nearly zero loss barriers, consistent with the proposed interpretation.
\end{abstract}

\section{Introduction}
In deep learning, mode connectivity has been argued to be an important geometrical condition for consistently coexisting several inference abilities in a large language model \cite{garipov2018,draxler2018,chen2025,theus2025}.
In classical machine learning, many studies have investigated mode connectivity heuristically in model parameter space.

In quantum machine learning (QML), mode connectivity in the loss landscape of parameterized quantum circuits has been observed experimentally, but its theoretical interpretation remains limited \cite{hamilton2022,hamilton2021}.
Thus, a theoretical understanding of mode connectivity in QML remains incomplete.

In this study, we do not aim to prove parameter-space mode connectivity in the classical or quantum conventional sense.
Instead, we demonstrate that mode connectivity in QML arises from the topology of low-loss sublevel sets on the dynamical Lie group, and overparameterization enables lifting of these Lie-group paths to parameter space.
We provide a theoretical interpretation of mode connectivity for parameterized unitaries generated by a set of generators.
In contrast to classical deep learning, where mode connectivity is usually formulated in parameter space, we formulate mode connectivity in QML in reachable unitary group generated by dynamical Lie algebra (DLA).
This is because the physically relevant object is the unitary transformation implemented by a parameterized unitary, while the parameters provide only a coordinate representation of this transformation.
Furthermore, toy numerical experiments are conducted to provide evidence consistent with our theoretical results.

Our main contributions of this study are as follows:
\begin{itemize}
\item We formulate mode connectivity for QML on the reachable unitary Lie group generated by the DLA, rather than directly in parameter space.
\item We prove a path-connectivity result for low-loss sublevel sets under a critical-value-free condition.
\item We clarify the role of overparameterization as a condition enabling local lifting of Lie-group paths to parameter space.
\item We provide toy numerical evidence based on geodesic interpolations between trained unitaries.
\end{itemize}

\section{Related work}
Hamilton et al. \cite{hamilton2021} numerically investigated mode connectivity in the loss landscape of parameterized quantum circuits.
Using three-qubit circuits with at most 13 trainable parameters, they applied the nudged elastic band method to identify piecewise low-loss paths between independently obtained minima.
Their results provided empirical evidence that connected low-loss regions can exist in the parameter space of parameterized quantum circuits.
However, the geometric or algebraic origin of such connectivity was not theoretically characterized.
In contrast, our study investigates mode connectivity from a Lie-theoretic perspective, focusing on the geometry of the reachable Lie group and the vector-space structure of its dynamical Lie algebra.
Thus, our study takes a step toward a theoretical framework for these empirical observations based on Lie-group geometry and dynamical Lie algebras.

\section{Method}
First, we introduce the following setup:\\
\noindent
{\bf Setup}
\begin{itemize}
\item Let a parameterized unitary be $ U(\bm{\theta}) = \prod_{j=1}^{L} e^{-i \theta_{j} G_{j}} $, where $ L $ denotes a depth (or number of parameters), $ \bm{\theta} = (\theta_{1}, \ldots, \theta_{L})^{\mathsf{T}} \in \bm{\Theta}_{L} \subset \mathbb{R}^{L} $, and $ G_{j} $ denotes a generator, where $ \mathsf{T} $ denotes transpose.
\item Lie group $ H $ (= $ \exp(\mathfrak{h}) $) is a compact connected Lie group, where $ \mathfrak{h} $ denotes the DLA generated by the generators.
\item Let $ \epsilon_{g} $ be the global minimum value of the loss function $ loss: \, H \rightarrow \mathbb{R} $.
\item DLA $ \mathfrak{h} $ is defined as $ \mathfrak{h} = \operatorname{Lie}\{ -iG_{1}, -iG_{2}, \ldots , -iG_{L} \} $, where $ \operatorname{Lie} $ denotes the Lie algebra generated by the enclosed operators, namely the smallest Lie algebra containing them.
\end{itemize}

Next, we define the mode connectivity as follows:\\
\begin{definition} [Mode connectivity]
  A low loss set $ M $ is given as follows:
  \begin{equation}
    M_{\epsilon_g + \epsilon} = \{ U \in \exp(\mathfrak{h}) \mid loss(U) \leq \epsilon_{g} + \epsilon \},
  \end{equation}
  where $ \epsilon > 0 $.
  For $ U_{a} \in M_{\epsilon_g + \epsilon} $ and $ U_{b} \in M_{\epsilon_g + \epsilon} $, if there exists a continuous path $ r: \, [0, 1] \rightarrow H $ such that $ r(0) = U_{a} $, $ r(1) = U_{b} $, and $ r(t) \in M_{\epsilon_g + \epsilon}$ ($ t \in [0, 1]$), $ M_{\epsilon_g + \epsilon} $ has mode connectivity.
  Here, $ r:[0, 1] \rightarrow H $.
\end{definition}

To construct the main theorem, we state a lemma and two assumptions.

\begin{lemma} [Existence of tangent descent directions on the reachable Lie group]\label{lem1}
  Let the boundary of $ M $ be $ \partial M_{\epsilon_g + \epsilon} \coloneqq \{ U \in H \mid loss(U) = \epsilon_{g} + \epsilon \} $.
  Assume that $ U $ is not a critical point of $ loss $ on $ H $.
  In a neighborhood of the boundary of $ M $, there exists a tangent direction $ \frac{d U}{d t} $ such that $ \frac{d loss(U)}{d t} = \braket{ \nabla_{H} loss(U), \frac{d U}{d t} }_{{\rm F}} < 0 $, where $ {\rm F} $ denotes the Frobenius inner product.
\end{lemma}
\begin{proof}
  Since $ U = \exp(\mathfrak{h}) $, $ \frac{d U}{d t} = U \omega $, $ U \omega \in T_{U} H $, and $ \omega \in \mathfrak{h} $.
  Here $ T_{U} $ is a tangent space of $ H $ in $ U $.
  When $ \frac{d U}{d t} = - \nabla_{H} loss(U) $, we have $ -\nabla_{H} loss(U) \in T_{U} H $, and for $ \omega \in \mathfrak{h} $, $ -\nabla_{H} loss(U) = U \omega $.
  Therefore, $ \braket{ \nabla_{H} loss(U), -\nabla_{H} loss(U) }_{{\rm F}} = -\| \nabla_{H} loss(U) \|^{2}_{{\rm F}}< 0 $, where $ \| \cdot \|_{\rm F} $ denotes the Frobenius norm.
\end{proof}

Based on Lemma \ref{lem1}, near the boundary of the low loss set, there exists a Lie algebra direction toward the inside of the set.
In addition, the Lie algebra direction is realized by a parameter velocity in the parameterized unitary.

For overparameterization \cite{larocca2023}, we state the following remark.
\begin{remark} [Relation to parameter-space connectivity]
  The Lie-group mode connectivity considered in this study does not necessarily imply global parameter-space connectivity, because the inverse image $ F_{L}^{-1}(r(t)) $ may have multiple connected components.
  However, if the map $ F_{L} $ is a submersion along a given Lie-group path and compatible local lifts can be chosen, then $ r(t) $ can admit a lift in parameter space, at least locally and piecewise smoothly.
  Thus, overparameterization is interpreted as a condition that makes the Lie-group connectivity visible in parameter space, rather than as the fundamental origin of connectivity itself.
\end{remark}
In QML, overparameterization should not be understood solely as having many parameters.
Rather, it is closely related to whether the parameter-to-unitary map has sufficient rank to cover tangent directions of the reachable Lie group, whose dimension is determined by the DLA.

We assume that the global minimizer set does not decompose into multiple disconnected orbits under the action of the dynamical Lie group $ H  = \exp(\mathfrak{h}) $.
\begin{assumption} [Connectedness of the global minimizer set]\label{ass1}
  We assume that the global minimizer $ M_{g} = \{ U \in H, loss(U) = \epsilon_{g} \} $ is connected.
\end{assumption}

\begin{remark} [A sufficient condition for Assumption \ref{ass1} is the single-orbit condition]
  If there exists a global minimizer $ U_{*} $ such that $ M_{g} = K U_{*} $, where $ K \subset H $ is a connected symmetry subgroup preserving the loss, then $ M_{g} $ is connected.
  Since $ H $ is connected and the orbit map $ \psi: \, H \rightarrow M_{g} $, $ \psi(V) = V U_{*} $, is continuous and surjective,
  $ M_{g} $ is the continuous image of a connected set.
\end{remark}

Here Assumption \ref{ass1} motivates the following assumption.

\begin{assumption} [Near-minimum connectedness]\label{ass2}
  For sufficiently small $ \delta > 0 $, $ M_{\epsilon_{g} + \delta} $ is path-connected.
\end{assumption}
Assumption \ref{ass2} is naturally satisfied, for example, when $ M_{g} $ is a connected nondegenerate minimum manifold.
We do not derive Assumption \ref{ass2} from first principles.
Rather, we regard it as a geometric hypothesis whose validity is expected to improve with increasing expressibility of the reachable Lie group.

Now, we construct the main result as a theorem.
\begin{theorem} [Lie-group mode connectivity in a critical-value-free low-loss band]\label{th1}
  Let $ H = \exp( \mathfrak{h}) $ be the compact connected reachable Lie group generated by DLA $ \mathfrak{h} $ of a parameterized unitary.
  Let $ loss: \, H \rightarrow \mathbb{R} $ be a smooth loss function.
  For $ c \geq \epsilon_{g} $, define the low-loss sublevel set $ M_{c} = \{ U \in H \mid loss(U) \leq c \} $.
  Assume that there exists $ \delta > 0 $ such that $ M_{\epsilon_{g} + \delta} $ is path-connected and $ loss $ has no critical points in the band $ loss^{-1}((\epsilon_{g} + \delta, \epsilon_{g} + \epsilon]) $.
  Then $ M_{\epsilon_{g} + \epsilon} $ is path-connected.
  Therefore, any two reachable unitaries $ U_{a} $ and $ U_{b} \in M_{\epsilon_{g} + \epsilon} $ can be connected by a continuous path $ r:[0, 1] \rightarrow H $ satisfying $ r(0) = U_{a} $, $ r(1) = U_{b} $, and $ r(t) \in M_{\epsilon_{g} + \epsilon} $ for all $ t \in [0, 1] $.
\end{theorem}
\begin{proof}
  Since $ H $ is compact and $ loss $ is smooth, the negative gradient flow of $ loss $ on $ H $ exists for all finite time.
  By assumption, $ loss $ has no critical points in the band $ loss^{-1}((\epsilon_{g} + \delta, \epsilon_{g} + \epsilon]) $.
  Therefore, by the deformation theorem (Theorem 3. 1 in \cite{milnor1963}) of Morse theory \cite{milnor1963,nicolaescu2011}, the sublevel set $ M_{\epsilon_{g} + \epsilon} $ deformation retracts onto $ M_{\epsilon_{g} + \delta} $.
  In particular, $ M_{\epsilon_{g} + \epsilon} $ and $ M_{\epsilon_{g} + \delta} $ have the same number of connected components.
  By the near-minimum connectedness assumption, $ M_{\epsilon_{g} + \delta} $ is path-connected.
  Hence $ M_{\epsilon_{g} + \epsilon} $ is path-connected.
  Therefore, any two reachable unitaries $ U_{a}, U_{b} \in M_{\epsilon_{g} + \epsilon} $, there exists a path $ r:[0, 1] \rightarrow H $ such that $ r(0) = U_{a} $, $ r(1) = U_{b} $, and $ r(t) \in M_{\epsilon_{g} + \epsilon} $.
\end{proof}

This theorem does not claim that the low-loss region in parameter space $ \bm{\Theta}_{L} $ is connected.
Instead, it shows that the physically relevant reachable unitaries form a connected low-loss region on $ H $.
In addition, we state the following remark.
\begin{remark}
  Even if $ M_{\epsilon_g + \epsilon} $ is path-connected, by the structure of $ F_{L}^{-1} $, there seems to exist several components in the parameter space.
  However, connectivity still depends on the topology of the sublevel sets and the absence of critical values in the relevant loss band.
\end{remark}

According to Theorem \ref{th1}, mode connectivity in QML does not arise in the case of a large number of parameters in a parameterized unitary.
However, when the low-loss sublevel set on the reachable Lie group generated by the DLA is connected, the low loss sublevel set is connected and the parameterized unitary is overparameterized to lift a low loss path on the Lie group to the parameter space, mode connectivity arises.
We conjecture that, for parameterized unitaries, in addition to overparameterization, implicit bias of optimization toward low-loss regions makes a situation of arising mode connectivity.

\section{Numerical experiments and results}\label{sec4}
To provide numerical evidence consistent with low-loss connectivity, we examined the loss barrier along geodesic interpolations between pairs of trained unitaries $ U_{a} $ and $ U_{b} $ using $ r(t) = U_{a} \exp(t \log(U_{a}^{\dagger} U_{b})) $.
In this study, we used barrier \cite{garipov2018} as a quantitative metric: $ B(U_{a}, U_{b}) = \max_{t} loss(r(t)) - \max(loss(U_{a}), loss(U_{b})) $.
Barrier is defined as the maximum loss along the interpolation path between two solutions minus the endpoint loss.
If the connectivity is preserved, the value of $ B $ becomes near zero.
We note that the principal logarithm used in the interpolation was not explicitly checked to lie in the DLA $ \mathfrak h $.
Hence, the numerical experiments should be interpreted as evidence for near-zero loss barriers under logarithmic unitary interpolation, rather than as a rigorous verification of DLA-reachable path connectivity.

First, for training parameterized unitaries, we generated input--output training data as follows.
The output state of a quantum circuit, which served as a target model, was obtained as $ \ket{\phi} = \left( \prod_{j=1}^{5} \exp(-i \theta_{j} G_{j}) \right) \, RY(x)^{\otimes 5} \, \ket{0}^{\otimes 5} $.
The corresponding target value was obtained as $ \hat{y} = \braket{\phi | \, Z(0) \, | \phi} $.
Here, the input value $ x $ was uniformly sampled from $ [0, 2\pi] $, and $ \theta_{j} $ was uniformly sampled from $ [-0.01, 0.01] $.
The generator $ G_{j} $ was sampled from the set of 5-qubit Pauli strings, excluding the all-identity string.
We repeated this sampling and evaluation procedure $ 20 $ times to obtain a training dataset of size $ 20 $.

For training, a parameterized unitary $ \prod_{j=1}^{L} \exp(-i \theta_{j} G_{j}) $ was used.
The generator $ G_{j} $ was also sampled from the set of 5-qubit Pauli strings, excluding the all-identity string.
The number of epochs was set to 500, and loss was measured by using the root mean square error (RMSE).
The optimizer was the simultaneous perturbation stochastic approximation (SPSA) \cite{gacon2021}.
The learning rate was 0.01, and the momentum term was 0.5.
The initial value of $ \bm{\theta} $ was uniformly sampled from $ [-0.1, 0.1] $.
We trained ten parameterized unitaries with ten different random seeds, varying $ L $ (the number of parameters) over $ \{5, 10, 20 \}$, and calculated $ B $ for 45 combinations of the trained unitaries.
The average and standard deviation results are shown in Table \ref{tab1}.
The parentheses indicate the standard deviation.
In addition, we calculated the DLA dimension of the generators.
Because our theoretical formulation concerns the low-loss geometry on the reachable Lie group $ H=\exp(\mathfrak{h}) $, rather than directly in the parameter space, we also computed the dimension of the DLA $ \mathfrak{h} $ for each randomly generated generators.
This quantity provides a diagnostic of the size of the reachable unitary manifold associated with the chosen Pauli generators.
\begin{table}[htb]
  \centering
  \caption{Results of loss of trained unitaries, $ B $, and DLA dimension}\label{tab1}
  \vspace{5pt}
  \begin{tabular}{cccc}\hline
    L (number of parameters) & Loss & B & $ {\rm dim}(\mathfrak{h}) $\\ \hline
    5 & 2.67E-04 (3.78E-04) & 4.42E-05 (7.86E-05) & 17.4 (4.82)\\
    10 & 4.62E-04 (5.43E-04) & 3.07E-05 (6.94E-05) & 388 (123.15)\\
    20 & 1.51E-03 (6.86E-04) & -1.35E-08 (2.15E-07) & 973.5 (148.50)\\ \hline
  \end{tabular}
\end{table}
In the table, the small negative value of $ B $ for $ L = 20 $ is attributed to numerical errors, since $ B \geq 0 $ by definition under exact evaluation and its magnitude is negligible compared with the loss scale.
The B values were much smaller than the loss values, and thus the B values are close to zero: $ B \approx 0 $.
Here the maximum possible DLA dimension for five qubits is $ {\rm dim} \mathfrak{su}(32) = 1023 $.
For $ L = 20 $, the DLA dimension reached 1023 for nine of the generated instances.
Increasing the number of parameters tends to promote overparameterization.

In addition, we show the results of $ loss(r(t)) $, where $ t \in [0, 1] $.
Figure~\ref{fig1} shows the average $ loss(r(t)) $.
The error bars indicate the standard deviation.
\begin{figure}[htbp]
  \centering
  \begin{minipage}{15cm}
    \SetFigLayout{2}{2}
    \subfigure[$L = 5$]{\includegraphics[width=7cm]{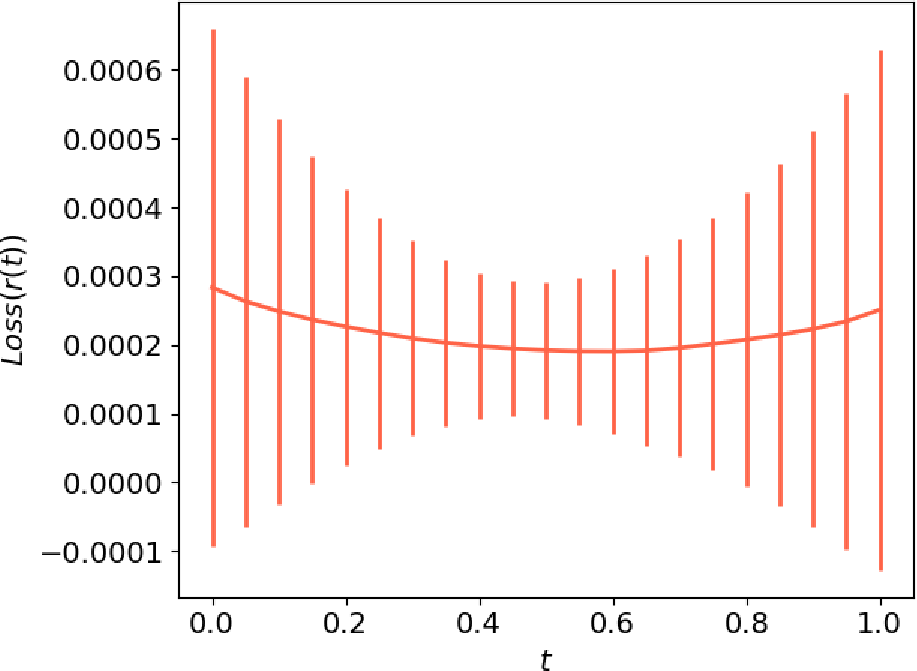}}
    \hfill
    \subfigure[$L = 10$]{\includegraphics[width=7cm]{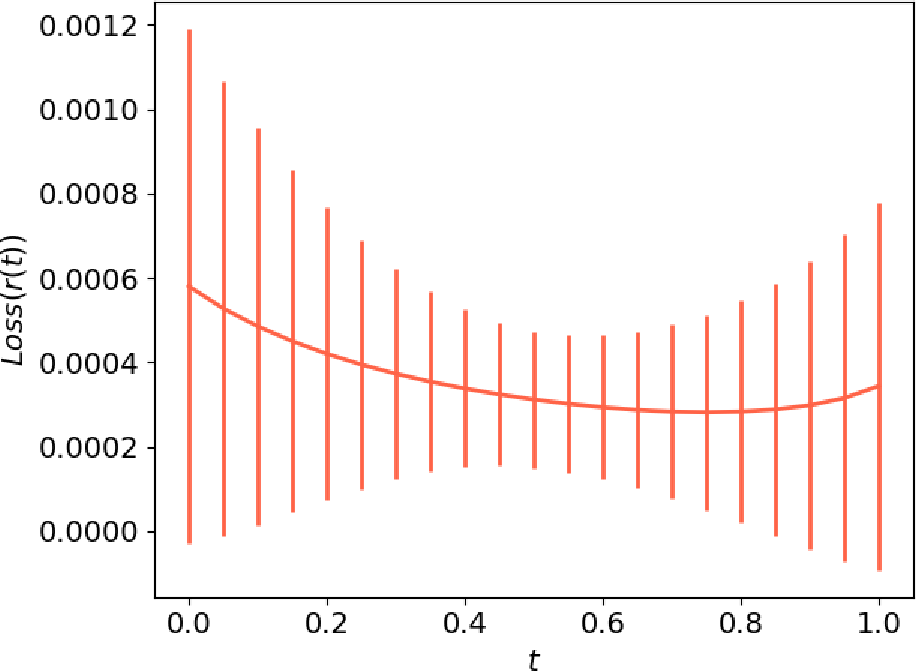}}\\
    \centering
    \subfigure[$L = 20$]{\includegraphics[width=7cm]{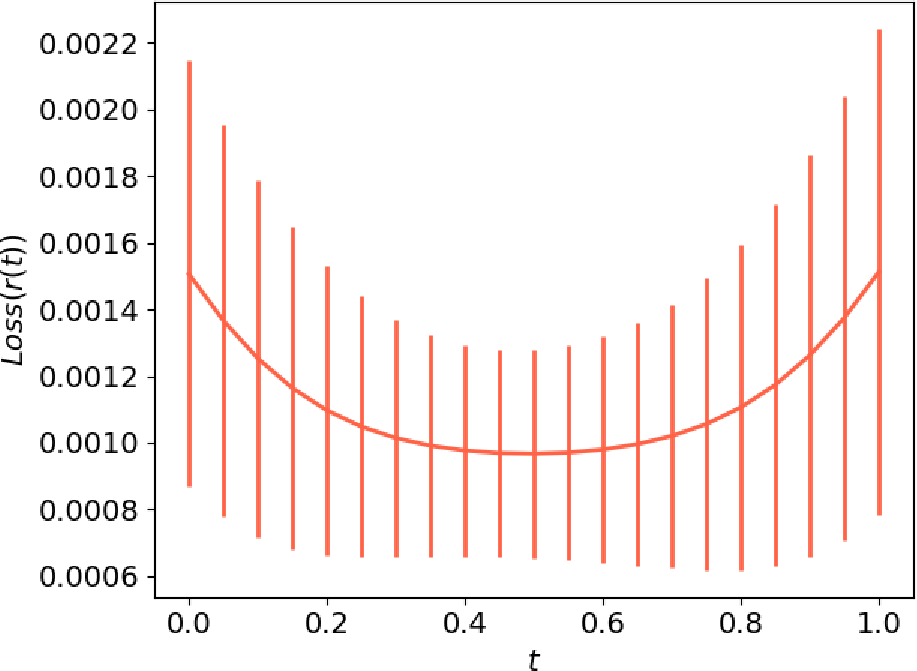}}
  \end{minipage}
  \caption{Results of $ loss(r(t)) $ for $ L \in \{5, 10, 20\} $}\label{fig1}
\end{figure}
As shown in the figures, $ loss(r(t)) $ values were not over the endpoint loss.
Although these experiments are limited to toy settings, the near-zero barriers provide empirical evidence consistent with the proposed Lie-group interpretation of mode connectivity.

\section{Conclusion}\label{sec5}
In this study, we theoretically analyzed mode connectivity of parameterized unitaries in QML from the standpoint of the DLA.
For mode connectivity to be observed in parameter space, overparameterization plays an important role by enabling local lifting of low-loss paths on the reachable Lie group.
Numerical experiments in a toy-problem setting were performed, and evidence supporting our theoretical results was obtained.
However, we think that further large-scale experiments are needed to confirm our results.
In future work, we should study the relationship between mode connectivity and the dimension of DLA, and barren plateaus.
We conjecture that increasing the dimension of DLA increases the likelihood that low-loss sublevel sets become connected.
In addition, alternative approaches based on symmetries of the DLA may provide further insight into mode connectivity and should be investigated in future work.

\bibliographystyle{plain}
\bibliography{references}

\end{document}